\documentclass[runningheads]{llncs}
\usepackage{cite}
\usepackage{amssymb,amsfonts,amsmath}
\usepackage{algorithmic}
\usepackage{graphicx}
\usepackage{textcomp}
\usepackage{graphicx}
\usepackage{subfigure}
\usepackage{booktabs, multirow}
\usepackage{algorithmic}
\usepackage{svg}
\usepackage{siunitx}
\usepackage{tikz}
\usetikzlibrary{arrows.meta, positioning}
\usepackage[ruled,vlined,linesnumbered]{algorithm2e}

\begin{document}
\title{SimCert: Probabilistic Certification for Behavioral Similarity in Deep Neural Network Compression}

\author{Jingyang Li\inst{1} \and Fu Song\inst{2}\and Guoqiang Li\inst{1}}

\institute{Shanghai Jiao Tong University, Shanghai 200240, China\\
\email{\{lijjjjjj, li.g\}@sjtu.edu.cn}
\and
Institute of Software, Chinese Academy of Sciences, Beijing 100190, China\\
\email{songfu@ios.ac.cn}
}

\maketitle

\begin{abstract}
Deploying Deep Neural Networks (DNNs) on resource-constrained embedded systems requires aggressive model compression techniques like quantization and pruning. However, ensuring that the compressed model preserves the behavioral fidelity of the original design is a critical challenge in the safety-critical system design flow. Existing verification methods often lack scalability or fail to handle the architectural heterogeneity introduced by pruning. In this work, we propose SimCert, a probabilistic certification framework for verifying the behavioral similarity of compressed neural networks. Unlike worst-case analysis, SimCert provides quantitative safety guarantees with adjustable confidence levels. Our framework features: (1) A dual-network symbolic propagation method supporting both quantization and pruning; (2) A variance-aware bounding technique using Bernstein’s inequality to tighten safety certificates; and (3) An automated verification toolchain. Experimental results on ACAS Xu and computer vision benchmarks demonstrate that SimCert outperforms state-of-the-art baselines.
\end{abstract}

\section{Introduction}

The deployment of Deep Neural Networks (DNNs) in safety-critical cyber-physical systems (CPS), such as autonomous vehicles and robotic controllers, requires rigorous guarantees on correctness and reliability~\cite{Mistry2022,Yin2022}. However, running large DNNs on edge devices is difficult due to limited memory and power. To solve this, developers use model compression techniques like quantization \cite{DBLP:journals/corr/HanMD15,DBLP:conf/iccv/GongLJLHLYY19} and pruning \cite{DBLP:conf/iclr/LiuSZHD19, DBLP:conf/cvpr/MolchanovMTFK19,Bachiri2024}. These steps reduce the model size and calculation costs, allowing complex models to run efficiently on resource-constrained hardware.

The tension between computational efficiency and functional preservation presents a core challenge in modern lightweight AI deployment. As noted by Matos Jr. et al. \cite{Matos2024}, relying solely on statistical accuracy measures is inadequate for safety-critical applications, as they fail to capture the network's vulnerability to specific inputs where performance may degrade significantly. Therefore, there is a critical need for formal guarantees of behavioral similarity across the input space.

Current verification approaches exhibit significant limitations in this context. First, prevailing methods primarily offer qualitative assessments (pass/fail outcomes) based on worst-case analysis. This requirement is often \textit{overly stringent} for practical applications: verification fails if even a single input triggers an unsafe output, regardless of how statistically rare such inputs may be~\cite{Matos2024}. In large-scale systems, discarding a highly efficient compressed model due to a negligible probability of failure represents a wasted optimization opportunity. Furthermore, most differential verification tools (e.g. ReluDiff~\cite{reludiff}, NeuroDiff~\cite{neurodiff}) rely on structural homogeneity assumptions, making them incompatible with architectural changes introduced by pruning~\cite{Zhong2023}.

To address this stringency, probabilistic verification offers a more practical alternative by providing statistical guarantees \cite{PROVEN}. However, a second layer of conservatism arises from the mathematical tools used. Existing frameworks typically rely on concentration inequalities like Hoeffding's inequality \cite{PROVEN}, which is \textit{variance-agnostic}—it bounds the error probability based solely on the value range. In the context of model compression, the behavioral difference between the original and compressed networks inherently exhibits \textit{low variance}. By neglecting this variance information, standard methods produce overly loose bounds. Consequently, they often fail to certify safe models (false rejections), effectively negating the benefit of relaxing the worst-case requirement.

Addressing these gaps requires solving two core technical challenges: (1) developing formal methods for \textit{quantitative} similarity verification that exploit the low-variance property of compression error, and (2) enabling verification across \textit{architecturally heterogeneous} networks resulting from pruning.

In this work, we propose \textbf{SimCert}, a probabilistic certification framework for verifying neural network behavioral similarity. We define similarity through quantitative thresholds, guaranteeing that for a given error tolerance $\epsilon$ and confidence level $\gamma$, the output difference remains within bounds.

SimCert makes the following contributions:
\begin{enumerate}
    \item \textbf{A Variance-Aware Probabilistic Certification Framework}: We propose SimCert, the first probabilistic framework that explicitly exploits the low-variance characteristic of compression errors. By integrating Bernstein’s inequality into dual-network symbolic propagation, we achieve significantly tighter certified radii compared to variance-agnostic methods (e.g., PROVEN~\cite{PROVEN}) without incurring significant computational overhead.

    \item \textbf{Unified Handling of Architectural Heterogeneity}: We introduce a structural alignment mechanism that enables the dual-network verification of architecturally heterogeneous models (e.g., pruned networks). Unlike existing differential verifiers~\cite{reludiff,neurodiff,QEBVerif} that are restricted to identical topologies, SimCert provides a unified solution for both quantization (parametric changes) and pruning (structural changes).

    \item \textbf{Automated Toolchain and Experimental Analysis}: We implement SimCert as an automated verification tool compatible with PyTorch. Extensive experiments on ACAS Xu, MNIST, and CIFAR-10 demonstrate that SimCert provides larger certified radii and higher confidence than state-of-the-art baselines, offering a robust solution for ensuring the reliability of embedded AI systems.
\end{enumerate}

\section{Preliminaries}\label{pre}
\subsection{Neural Networks}

A neural network consists of an input layer, several hidden layers, and an output layer. These layers can perform either linear or non-linear operations. Fully connected layers are a typical example of linear layers,
while activation layers are used for non-linear operations, such as the ReLU function.
\begin{figure}[!ht]
    \centering
    \includegraphics[width=0.95\linewidth]{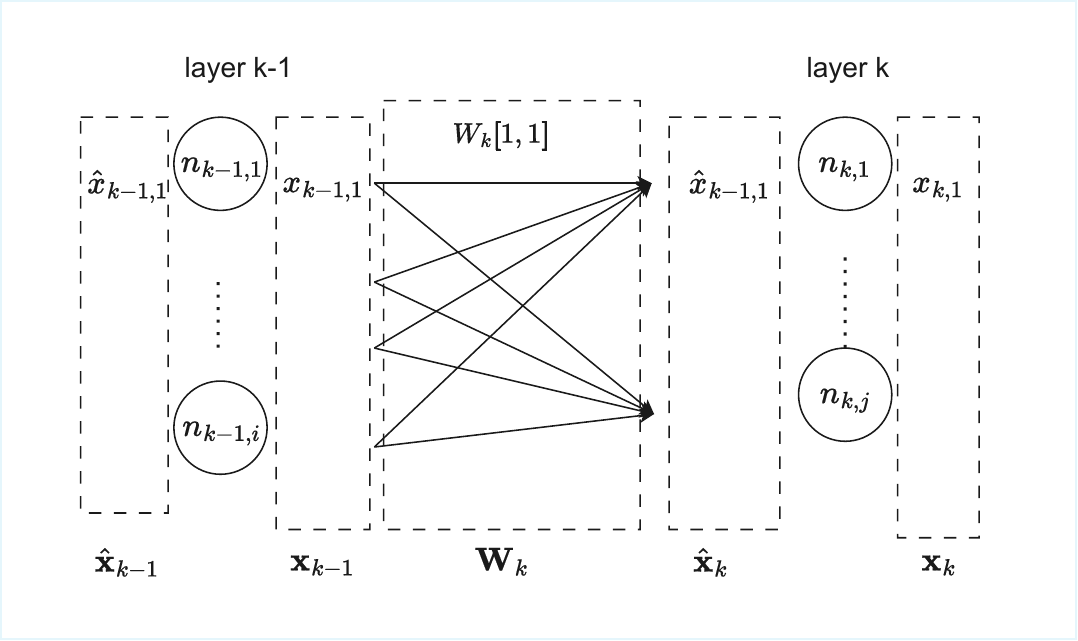}
    \caption{Diagram of network notations}
    \label{fig:diagram of network notation}
\end{figure}

In our paper, we use $[L]$ to represent the set of strictly positive integers up to $L$. As depicted in Figure \ref{fig:diagram of network notation}, $\mathbf{W}_k$ denotes the weight matrix of the $k$-th layer, $n_{k,j}$ refers to the $j$-th node in the $k$-th layer, $x_{k,j}$ refers to the input of the $j$-th node in the $k$-th layer, $\hat{x}_{k,j}$ refers to the output of the $j$-th node in the $k$-th layer, and $W_k[i,j]$ represents the weight of the edge from $n_{k-1,i}$ to $n_{k,j}$. We use bold symbols for vectors and matrices, like $\mathbf{x}_k$ and $\mathbf{W}_k$, and regular symbols for scalars, such as $x_{k-1,i}$ and $W_k[i,j]$.

Using this notation, we define an ReLU neural network with $L$ layers through its weight matrices $\mathbf{W}_k$ and bias vectors $\mathbf{b}_k$, where $k$ ranges from $1$ to $L$. We index the input layer as layer $0$ by convention. The neural network computes the output for an input $\hat{\mathbf{x}}_0$ from a bounded input domain $\mathcal{X}$ by recursively applying linear layers: $\mathbf{x}_k = \mathbf{W}_k\hat{\mathbf{x}}_{k-1}+\mathbf{b}_k$, and ReLU layers: $\hat{\mathbf{x}}_k=\max(\mathbf{0},\mathbf{x}_k)$ until the output $\mathbf{x}_L$ is obtained.

\subsection{Neural Network Compression Techniques}
Large-scale DNNs (Deep Neural Networks) are limited by their size and the computational cost, making it difficult to deploy them on edge devices. Therefore, a series of network compression methods have been derived.

One important type of approaches is model quantization~\cite{DBLP:journals/corr/HanMD15,DBLP:conf/iccv/GongLJLHLYY19}. Model quantization is a way to compress the parameters of a network by representing the parameters (weights) and feature maps (activations) of a dense neural network, originally represented as floating-point values, using fixed-point (integer) representation. Finally, the fixed-point data is dequantized back into floating-point data to obtain the result. Model quantization can speed up model inference and reduce model memory usage. Model quantization is based on the deep network's tolerance to noise, as it essentially introduces a certain amount of noise (quantization error) to the deep network. If the number of quantization bits is appropriate, model quantization will not lead to significant loss in accuracy.

Another important type of compression techniques for neural networks are network pruning~\cite{DBLP:conf/nips/ChoAN23,DBLP:conf/iclr/LiuSZHD19,DBLP:conf/cvpr/MolchanovMTFK19}. Neural network pruning aims to remove redundant parts of a network that consume a significant amount of resources while maintaining good performance. Despite the remarkable learning capacity of large neural networks, not all parts of the network are actually useful after the training process. The idea behind neural network pruning is to eliminate these unused parts without affecting the network's performance. Formally, given a neural network $f$ and its pruned version $f'$, network pruning is equivalent to the optimization problem:
\begin{align*}
    \max_{f'}\quad &Acc(f',D) \\
     \text{s.t. } \quad &Budget(f\leftarrow f')\leq C,
\end{align*}
where $Acc(f',D)$ represents the accuracy of pruned network $f'$ on dataset $D$ and $Budget(f\leftarrow f')$ represents the budget of the pruning process.

It should be noted that the majority of current quantization techniques~\cite{Xie-2506-17870,HuangS0LHWXC22,BhalgatLNBK20,llmint8} do not modify the network architectures, thereby enabling the formal verification of accuracy loss in quantization through existing differential verification tools. Conversely, with existing network pruning techniques, empirical evaluation of accuracy loss predominates, lacking formal guarantees. Unfortunately, existing differential verification tools do not currently support such scenarios.

\subsection{Qualitative Similarity Certification}

\begin{figure}[!ht]
  \centering
  \subfigure[Original neural network]{
    \includegraphics[width=0.9\columnwidth]{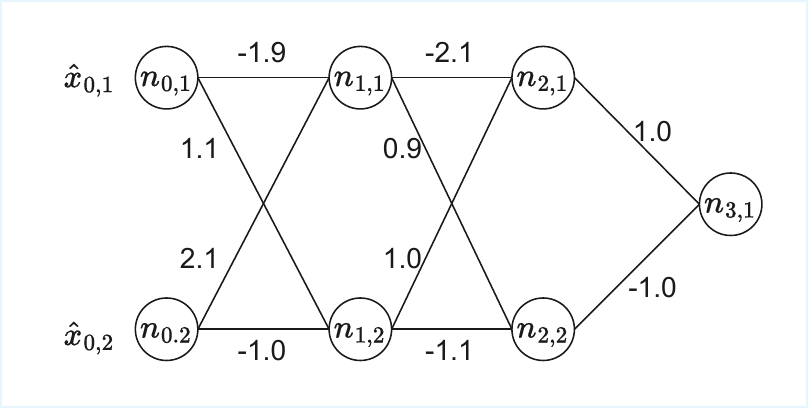}
    \label{fig:f}
  }~

  \subfigure[Quantized neural network]{
    \includegraphics[width=0.9\columnwidth]{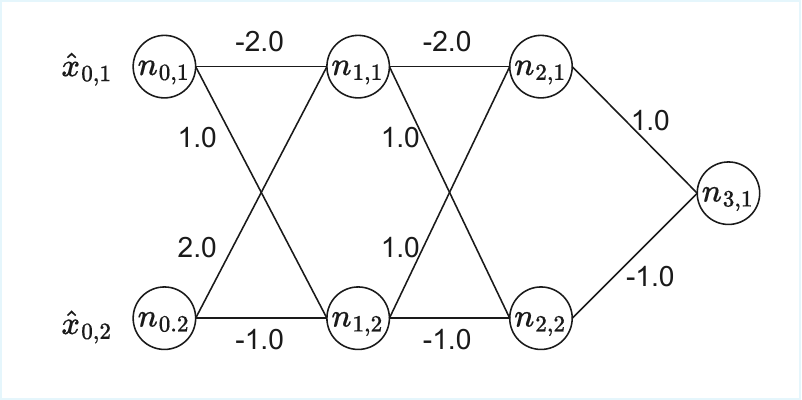}
    \label{fig:f'}
  }~

  \caption{The pair of neural networks of the original one and the quantized one. }
  \vspace{-3mm}
  \label{fig:network_pair}
\end{figure}
Qualitative similarity certification establishes behavioral equivalence between neural networks during compression (as illustrated in Figure~\ref{fig:network_pair}), a process also known as differential verification~\cite{reludiff,neurodiff}. This approach formally characterizes the relationship between two networks $f$ (original) and $f'$ (compressed) through the following definition:

\begin{definition}[Qualitative $\epsilon$-Similarity Certification]\label{qsc}
Given a pair of neural networks $\langle f,f'\rangle:\mathcal{X}\subseteq\mathbb{R}^{n}\rightarrow\mathbb{R}^{m}$ where $\mathcal{X}$ is a compact input domain, and without loss of generality assuming $m=1$, for predefined error threshold $\epsilon>0$, we say $f'$ passes $\epsilon$-similarity certification relative to $f$ on $\mathcal{X}$ if:
\begin{equation*}
\forall x\in\mathcal{X},\ |f(x)-f'(x)| \leq \epsilon.
\end{equation*}
\end{definition}

\begin{figure}[h!]
    \centering
    \includegraphics[width=0.9\linewidth]{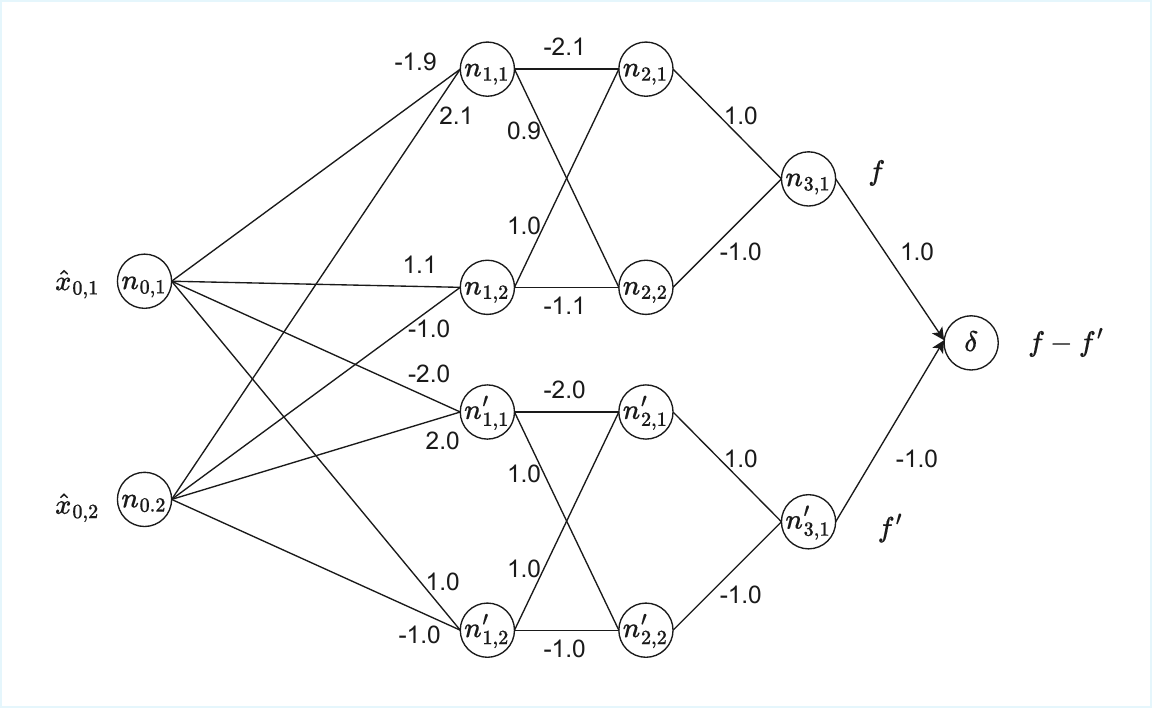}
    \caption{Naive encoding to compute the difference of the two neural networks}
    \label{fig:naive}
\end{figure}

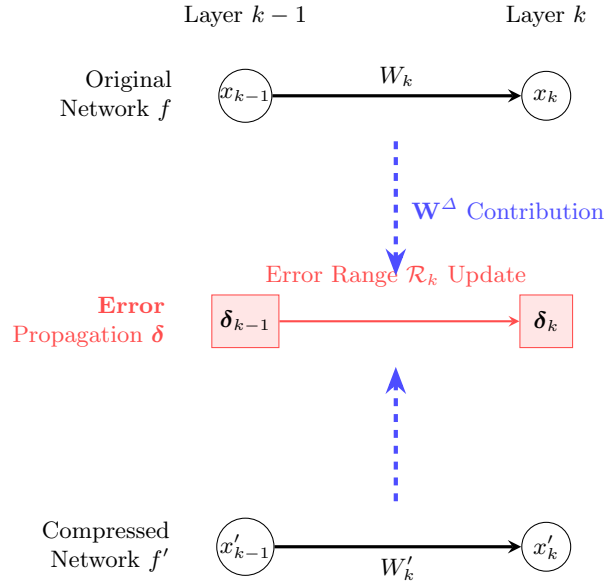
\begin{figure}[!t]
\centering
\begin{tikzpicture}[
    node distance=2.2cm, 
    layer/.style={rectangle, draw, minimum size=0.8cm},
    data/.style={circle, draw, minimum size=0.65cm, inner sep=0pt, align=center},
    error/.style={rectangle, draw=red!70, fill=red!10, minimum size=0.7cm, align=center, font=\small},
    propagation/.style={->, >=stealth, very thick, black},
    error_prop/.style={->, >=stealth, thick, red!70},
    weight_diff/.style={->, >=stealth, dashed, shorten >=3pt, shorten <=3pt, ultra thick, blue!70} 
]


\def\yOrig{3.0} 
\def\yError{0}  
\def\yComp{-3.0} 
\def\xPos{4.0} 

\node[data] (xkm1) at (0, \yOrig) {$x_{k-1}$};
\node[error] (deltakm1) at (0, \yError) {$\boldsymbol{\delta}_{k-1}$};
\node[data] (xpkm1) at (0, \yComp) {$x'_{k-1}$};

\node[data] (xk) at (\xPos, \yOrig) {$x_k$};
\node[error] (deltak) at (\xPos, \yError) {$\boldsymbol{\delta}_k$};
\node[data] (xpk) at (\xPos, \yComp) {$x'_k$};


\draw[propagation] (xkm1) -- (xk) node[midway, above] {$W_k$};

\draw[propagation] (xpkm1) -- (xpk) node[midway, below] {$W'_k$};

\draw[error_prop] (deltakm1) -- (deltak);
\node[red!70] at (\xPos/2, 0.6) {\small Error Range $\mathcal{R}_k$ Update};

\draw[weight_diff, -{Stealth}] (\xPos/2, \yOrig - 0.5) -- (\xPos/2, \yError + 0.5) node[midway, right, align=center, blue!70, xshift=2pt] {$\mathbf{W}^\Delta$ Contribution};
\draw[weight_diff, -{Stealth}] (\xPos/2, \yComp + 0.5) -- (\xPos/2, \yError - 0.5);


\node[anchor=south] at (0, \yOrig + 0.8) {Layer $k-1$};
\node[anchor=south] at (\xPos, \yOrig + 0.8) {Layer $k$};

\node[left=0.5cm of xkm1, align=right] (label_orig) {Original\\Network $f$};
\node[left=0.5cm of deltakm1, align=right, red!70] (label_error) {\textbf{Error}\\Propagation $\boldsymbol{\delta}$};
\node[left=0.5cm of xpkm1, align=right] (label_comp) {Compressed\\Network $f'$};

\end{tikzpicture}
\caption{\label{fig:dual_network_propagation} Conceptual illustration of Dual-Network Symbolic Error Propagation.}
\end{figure}

A naive approach encodes similarity constraints by composing $f$ and $f'$ into a single feed-forward network (Figure~\ref{fig:naive}), equivalent to computing $f - f'$. While theoretically verifiable using sound and complete techniques, this method exhibits significant value range inflation in the output layer during worst-case analysis.

Methods like ReluDiff~\cite{reludiff} and NeuroDiff~\cite{neurodiff} overcome this limitation by exploiting structural and behavioral similarities. As demonstrated in Figure~\ref{fig:dual_network_propagation}, these approaches adopt a \textbf{Dual-Network Symbolic Propagation} mechanism, simultaneously tracking the original network ($f$) and the compressed network ($f'$). The core principle involves:

\begin{itemize}
    \item \textbf{Layer-wise Differencing:} The method explicitly models the difference between intermediate outputs, denoted by the error term $\boldsymbol{\delta}_k = x_k - x'_k$. This is achieved without separating calculations for $f(\mathbf{x})$ and $f'(\mathbf{x})$.
    \item \textbf{Structural Exploitation:} As the figure shows, the symbolic propagation of the error $\boldsymbol{\delta}_{k-1}$ to $\boldsymbol{\delta}_k$ is driven by the weight difference $\mathbf{W}^\Delta = \mathbf{W}_k - \mathbf{W}'_k$. This allows the verification to leverage the strong structural correlation ($||\mathbf{W}^\Delta||$ is small) to achieve tighter bounds.
    \item \textbf{Range Refinement:} Symbolic methods employ refinement processes (e.g., using linear relaxations) to tighten the error range $\mathcal{R}_k$ (representing $\boldsymbol{\delta}_k$) and employ refinement processes (e.g., using linear relaxations) to tighten these bounds.
\end{itemize}
This differential propagation strategy significantly improves verification efficiency and accuracy by effectively bounding the worst-case error accumulation based on network similarity.

\subsection{Probabilistic Verification}
Traditional formal verification methods focus on deterministic guarantees under worst-case scenarios, whereas practical applications often require robustness evaluation under probabilistic distributions. Probabilistic verification introduces probability measures to establish a quantitative analysis framework for neural network behavior under stochastic perturbations.
Notable approaches include the \textbf{PROVEN} framework~\cite{PROVEN}, which certifies distribution-aware robustness guarantees under specified noise distributions, and hybrid methods~\cite{ProbabilisticPracticalRobustness} combining abstract interpretation with Monte Carlo sampling to balance efficiency and probabilistic assurances for practical deployment. Statistical evaluation techniques~\cite{StatisticalRobustnessNN} leverage adaptive sampling to estimate perturbation tolerance while quantifying violation probabilities, complementing qualitative verification tools. Algorithmic advancements like the branch-and-bound PV method~\cite{PVBranchBound} iteratively refine probability bounds through dynamic branching, achieving higher verification success rates. For quantitative risk assessment, probabilistic star-based reachability analysis~\cite{QuantitativeVerificationProbStars,TranCLLOHF25} enables efficient uncertainty propagation modeling in safety-critical systems. These methods collectively enhance verification precision, scalability, and applicability across diverse scenarios. However, none address the critical challenge of similarity certification between original and compressed neural networks in model compression settings.

\section{Problem Formulation}\label{problem}

\subsection{$\epsilon$-$\gamma$ similarity certification}
Consider a feed-forward neural network $f$ with input layer, two hidden layers, and output layer. After transformations like fine-tuning, quantization, or model transfer, we obtain a modified network $f'$. A core concern for developers lies in assessing how the transformation alters $f$'s behavior: while the modified $f'$ may gain practical advantages (e.g., faster inference from quantization, better domain adaptation from transfer), it could also introduce subtle behavioral deviations that risk failing application requirements.

Traditional qualitative verification~\cite{reludiff,neurodiff,QEBVerif} enforces strict behavioral equivalence through $\epsilon$-similarity certification, requiring:
\[
\forall x \in \mathcal{X}: |f(x)-f'(x)| \leq \epsilon
\]
This requirement, however, can be overly stringent for practical applications. First, its universal quantifier means certification fails if even a single input $x\in\mathcal X$ satisfies $|f(x)-f'(x)|>\epsilon$, regardless of how statistically rare such inputs may be. Second, designers of neural network systems typically assert statistical behavioral guarantees, claiming systems satisfy target properties with high probability rather than universally~\cite{BalutaSSMS19}. Third, qualitative verification provides no insight into the scale of deviation. For example, two compressed networks $f'_1$ and $f'_2$ might both fail strict $\epsilon$-similarity, but $f'_1$ deviates for 0.5\% of inputs while $f'_2$ deviates for 10\%. This is a critical distinction for risk assessment, which qualitative methods cannot capture.

To address this limitation, we introduce probabilistic $\epsilon$-$\gamma$ similarity certification, which relaxes the universal quantifier while providing statistical guarantees. Unlike traditional qualitative checks, our certification focuses on the frequency of inputs where $f$ and $f'$ satisfy $\epsilon$-similarity: it requires that at least $\gamma$ proportion of inputs in $\mathcal{X}$ meet $|f(x)-f'(x)| \leq \epsilon$. The notion of $\epsilon$-$\gamma$ similarity certification is formally defined as follows:

\begin{definition}[$\epsilon$-$\gamma$ Similarity Certification]
Given a pair of neural networks $\langle f,f'\rangle:\mathcal{X}\subseteq\mathbb{R}^{n}\rightarrow\mathbb{R}$ where $\mathcal{X}$ is a compact input domain, for predefined error threshold $\epsilon>0$ and confidence level $\gamma\in(0,1)$, we say $f'$ passes $\epsilon$-$\gamma$ similarity certification relative to $f$ on a measurable subset $\mathcal{C}\subseteq\mathcal{X}$ if:
\begin{equation*}
\frac{\mu(\mathcal{C})}{\mu(\mathcal{X})} \geq \gamma \quad \text{and} \quad \forall x\in\mathcal{C},\ |f(x)-f'(x)| \leq \epsilon
\end{equation*}
where $\mu(\cdot)$ denotes the Lebesgue measure.
\end{definition}
Here, the Lebesgue measure~\cite{tao2011introduction,schilling2017measures} generalizes the concept of (1) Length in 1D ($\mathbb{R}$); (2) Area in 2D ($\mathbb{R}^2$); (3) Volume in higher dimensions ($\mathbb{R}^n$), thus quantifying the "size" of sets in Euclidean space.

\subsection{Analytical Setup: Dual-Network Notation and Certified Radius}
To facilitate the certification of behavioral similarity, we employ a dual-network symbolic propagation mechanism. This mechanism is crucial as it explicitly models the difference between $f$ and $f'$ activations, which is more precise than analyzing the difference network $f-f'$ separately (as discussed in Section~\ref{approach}).

\noindent\textbf{Dual-Network Notation.}
Let $W_k$ and $b_k$ denote the weight matrix and bias vector of the $k$-th layer in $f$. Similarly, $W'_k$ and $b'_k$ denote those in $f'$. We define the layer-wise deviations as:
$$W^{\Delta}_k = W_k - W'_k, \quad b^{\Delta}_k = b_k - b'_k$$
During the symbolic propagation, we track the pre-activation $x_{k,j}$ and post-activation $\hat{x}_{k,j}$ of neurons in $f$. We also explicitly model the accumulated error $\delta_{k,j}$ between the two networks. The layer-wise pre-activation error is defined as $\delta_k = x_k - x'_k$, and the post-activation error is $\hat{\delta}_k = \hat{x}_k - \hat{x}'_k$. This explicit modeling of $W^\Delta$ is central to SimCert, allowing us to leverage structural correlations to cancel out common signals and avoid value range inflation.

\begin{definition}(Certified Radius)\label{def:certified_radius}
    Given a reference input $x_0$, a confidence level $\gamma$, and a similarity threshold $\epsilon$, the \textbf{Certified Radius} $R(\epsilon, \gamma)$ is defined as the maximum radius $r$ of an $\ell_\infty$-ball centered at $x_0$ such that the probabilistic similarity condition holds:
$$R(\epsilon, \gamma) = \max \{ r \mid \mathbb{P}_{x \sim \mathcal{D}_{x_0, r}}(|f(x) - f'(x)| \le \epsilon) \ge \gamma \}$$
where $\mathcal{D}_{x_0, r}$ represents the input distribution (e.g., uniform) within $B_\infty(x_0, r)$.
\end{definition}
This metric quantifies the maximum operational envelope (radius $r$) for which the compressed model is statistically guaranteed to remain consistent with the original model with at least $\gamma$ probability.

\subsection{Technical Challenges and Proposed Solutions}

\noindent\textbf{C1: Structural Correspondence Problem via Zero-Padding.}
Building on existing definitions~\cite{reludiff,neurodiff}, the dual-network symbolic propagation assumes identical architectures to compute layer-wise differences. This mechanism fails for networks transformed by pruning, which alters the network topology. For instance, when pruning removes neuron $j$ in layer $k$, the subsequent layer's input dimensions are mismatched. The dual-network propagation expects layer-wise correspondence between activations ($\hat{x}_k$ vs $\hat{x}'_k$) to compute their difference $\hat{\delta}_k$. Without alignment, error propagation fails at layer $k$ due to incompatible tensor dimensions. To resolve this, we introduce \textbf{zero-padding} (detailed in Section~\ref{approach}) to restore structural equivalence: for a pruned neuron $j$ in layer $k$, we inject a zero-valued activation at position $j$ in $f'$. This dummy neuron maintains positional alignment while ensuring zero contribution to subsequent layers ($0 \cdot W'_{k}[\cdot,j] = 0$).

\noindent\textbf{C2: Suboptimal Bounds from Hoeffding's Inequality.}
Current probabilistic verification approaches, such as PROVEN~\cite{PROVEN}, yield suboptimal bounds when adapted to network similarity certification by relying on Hoeffding's inequality. Hoeffding's inequality neglects input variance and produces loose estimates, failing to account for the concentration of the input distribution. This can lead to substantially less informative probability lower bounds. To overcome this, we derive tighter probabilistic bounds using \textbf{Bernstein's inequality}, which is variance-aware, allowing for tighter similarity certification bounds for the same confidence level.

\noindent\textit{Illustrative Example.}
We use a simple example to illustrate the advantage of our approach over existing Hoeffding's inequality based computation. Consider a one-dimensional uniform input distribution
\(
X \sim \mathrm{Uniform}(-1, 1),
\)
with linear relaxations of the network output difference given by
\(
\delta^L(x) = 0.09x - 0.45\) and \( \delta^U(x) = 0.03x + 0.40.
\)
Our objective is to evaluate the similarity probability
\(
\mathbb{P}\left(|\delta(X)| \le 0.5\right).
\)
Under this configuration, the Hoeffding-based approach yields a probability lower bound of
\(
\text{LB}_{\text{Hoeffding}} = 0.139,
\)
while our variance-aware Bernstein-based method achieves a tighter probabilistic bound of
\(
\text{LB}_{\text{Bernstein}} = 0.236,
\)
representing a 11.3\% improvement relative to Hoeffding’s bound. Critically, tighter probability bounds enable larger certified radii for the same confidence level. We detail the derivation and the theoretical advantage in Section~\ref{approach}.

\section{Probabilistic Similarity Certification: SimCert}\label{approach}
\subsection{Dual-Network Symbolic Propagation}
SimCert leverages prior dual-network symbolic propagation~\cite{reludiff,neurodiff} to simultaneously track behavioral characteristics of the original network $f$ and its compressed counterpart $f'$. This approach explicitly models weight perturbations from compression while ensuring reliable error propagation. By employing convex relaxation to linearly bound ReLU activations, the method achieves efficient inter-layer propagation. We formally introduce this framework and describe our zero-padding mechanism to extend it to heterogeneous network pairs.


The symbolic execution maintains following state vectors (using the deviation notations defined in Section 3.2):

1. \textbf{Normal activation vectors}: $\mathbf{x}_k \in \mathbb{R}^{n_k}$ (pre-activation) and $\hat{\mathbf{x}}_k \in \mathbb{R}^{n_k}$ (post-activation) for layer $k$ of $f$, where $n_k \in \mathbb{N}$ is the neuron count.

2. \textbf{Error accumulation vectors}: $\boldsymbol{\delta}_k \in \mathbb{R}^{n'_k}$ (pre-activation) and $\hat{\boldsymbol{\delta}}_k \in \mathbb{R}^{n_k}$ (post-activation) for layer $k$, encoding approximation errors from compression.

The goal of dual-network symbolic propagation is to compute linear transformations that bound the accumulated error at the output layer. This is achieved by deriving layer-wise linear approximations that constrain four key quantities: the original network's activations ($\mathbf{x}_k$), the compressed network's activations ($\hat{\mathbf{x}}_k$), and their respective accumulated errors ($\boldsymbol{\delta}_k$ and $\hat{\boldsymbol{\delta}}_k$).

Specifically, for each layer $k$, we compute coefficient matrices and bias vectors that satisfy the following bounding inequalities for any input $\mathbf{x} \in \mathbb{R}^{d_{\text{in}}}$:
\begin{align*}
\mathbf{A}^l_k \mathbf{x} + \mathbf{b}^l_k &\leq \mathbf{x}_k \leq \mathbf{A}^u_k \mathbf{x} + \mathbf{b}^u_k \\
\hat{\mathbf{A}}^l_k \mathbf{x} + \hat{\mathbf{b}}^l_k &\leq \hat{\mathbf{x}}_k \leq \hat{\mathbf{A}}^u_k \mathbf{x} + \hat{\mathbf{b}}^u_k \\
\mathbf{C}^l_k \mathbf{x} + \mathbf{d}^l_k &\leq \boldsymbol{\delta}_k \leq \mathbf{C}^u_k \mathbf{x} + \mathbf{d}^u_k \\
\hat{\mathbf{C}}^l_k \mathbf{x} + \hat{\mathbf{d}}^l_k &\leq \hat{\boldsymbol{\delta}}_k \leq \hat{\mathbf{C}}^u_k \mathbf{x} + \hat{\mathbf{d}}^u_k,
\end{align*}
where $\mathbf{A}^l_k, \mathbf{A}^u_k, \hat{\mathbf{A}}^l_k, \hat{\mathbf{A}}^u_k \in \mathbb{R}^{n_k \times d_{\text{in}}}$ are bounding matrices for activation bounds, $\mathbf{C}^l_k, \mathbf{C}^u_k, \hat{\mathbf{C}}^l_k, \hat{\mathbf{C}}^u_k \in \mathbb{R}^{n_k \times d_{\text{in}}}$ are bounding matrices for error bounds, $\mathbf{b}^l_k, \mathbf{b}^u_k, \hat{\mathbf{b}}^l_k, \hat{\mathbf{b}}^u_k \in \mathbb{R}^{n_k}$ are bounding vectors for activations, $\mathbf{d}^l_k, \mathbf{d}^u_k, \hat{\mathbf{d}}^l_k, \hat{\mathbf{d}}^u_k \in \mathbb{R}^{n_k}$ are bounding vectors for errors, and superscripts $l$ and $u$ denote lower and upper bounds respectively.
Through iterative layer-wise propagation, we ultimately obtain linear transformations that bound the accumulated error at the output layer. Specifically, we derive bounding matrices $\mathbf{C}^l_L, \mathbf{C}^u_L \in \mathbb{R}^{n_L \times d_{\text{in}}}$ and bounding vectors $\mathbf{d}^l_L, \mathbf{d}^u_L \in \mathbb{R}^{n_L}$ satisfying:
\begin{align*}
\mathbf{C}^l_L \mathbf{x} + \mathbf{d}^l_L &\leq \boldsymbol{\delta}_L \leq \mathbf{C}^u_L \mathbf{x} + \mathbf{d}^u_L
\end{align*}
This formulation constrains the true error $\boldsymbol{\delta}_L$ by linear functions of the network input $\mathbf{x}$, where the coefficients capture cumulative effects of weight deviations and approximation errors from preceding layers. These linear bounds allow direct computation of output error bounds for arbitrary input regions.

\noindent\textbf{Initialization. }The symbolic propagation begins at the input layer ($k=0$) with exact initialization:
\begin{equation*}
\hat{\mathbf{x}}_0 = \mathbf{x}, \quad \hat{\boldsymbol{\delta}}_0 = \mathbf{0}
\label{eq:init-states}
\end{equation*}
where $\mathbf{0} \in \mathbb{R}^{d_{in}}$ is the zero vector. This establishes $\hat{\mathbf{x}}_0$ as the original input and $\hat{\boldsymbol{\delta}}_0$ as zero initial error.
For the first hidden layer ($k=1$), we initialize tight linear bounds before activation functions:
\begin{align*}
\mathbf{A}^l_1 = \mathbf{A}^u_1 &= \mathbf{W}_1, & \mathbf{b}^l_1 = \mathbf{b}^u_1 &= \mathbf{b}_1 \\
\mathbf{C}^l_1 = \mathbf{C}^u_1 &= \mathbf{W}_1 - \mathbf{W}'_1, & \mathbf{d}^l_1 = \mathbf{d}^u_1 &= \mathbf{b}_1 - \mathbf{b}'_1
\end{align*}
where $\mathbf{W}_1, \mathbf{W}'_1 \in \mathbb{R}^{n_1 \times n_0}$ denote the weights of the first layer in $f$ and $f'$, $\mathbf{b}_1,\mathbf{b}'_1 \in \mathbb{R}^{n_1}$ denote the bias of the first layer in $f$ and $f'$, and $n_0 = n_0 = d_{\text{in}}$.

\noindent\textbf{Error Propagation. }
For each network layer $k \in \{1,\ldots,L\}$, the propagation process comprises two operations: linear transformation and ReLU activation function application. The linear transformation is mathematically expressed as:
\begin{align*}
        \mathbf{x}_k &= \mathbf{\mathbf{W}_k\hat{x}}_{k-1} + \mathbf{b}_k \\
        \boldsymbol{\delta}_k &= \mathbf{W}^\Delta_k\mathbf{\hat{x}}_{k-1} + \mathbf{W}'_k\boldsymbol{\hat{\delta}}_{k-1} + (\mathbf{b}_k - \mathbf{b}'_k),
\end{align*}
where $\mathbf{W}^\Delta_k = \mathbf{W}_k - \mathbf{W}'_k$ represents the weight deviation between networks.
The linear bounding coefficients are then updated through matrix compositions:
\begin{align*}
\mathbf{A}^l_k &= \mathbf{W}_k \hat{\mathbf{A}}^l_{k-1} ,\quad\mathbf{b}^l_k = \mathbf{W}_k \hat{\mathbf{b}}^l_{k-1} + \mathbf{b}_k \\
\mathbf{A}^u_k &= \mathbf{W}_k \hat{\mathbf{A}}^u_{k-1} ,\quad \mathbf{b}^u_k = \mathbf{W}_k \hat{\mathbf{b}}^u_{k-1} + \mathbf{b}_k \\
\mathbf{C}^l_k &= \mathbf{W}^\Delta_k \hat{\mathbf{A}}^l_{k-1} + \mathbf{W}'_k \hat{\mathbf{C}}^l_{k-1}\\
\mathbf{d}^l_k &= \mathbf{W}^\Delta_k \hat{\mathbf{b}}^l_{k-1} + \mathbf{W}'_k \hat{\mathbf{d}}^l_{k-1} + (\mathbf{b}_k - \mathbf{b}'_k)  \\
\mathbf{C}^u_k &= \mathbf{W}^\Delta_k \hat{\mathbf{A}}^u_{k-1} + \mathbf{W}'_k \hat{\mathbf{C}}^u_{k-1}\\
\mathbf{d}^u_k &= \mathbf{W}^\Delta_k \hat{\mathbf{b}}^u_{k-1} + \mathbf{W}'_k \hat{\mathbf{d}}^u_{k-1} + (\mathbf{b}_k - \mathbf{b}'_k)
\end{align*}
This updates mechanism propagates bounding expressions through linear layers via direct matrix multiplications. It maintains linear relationships between network inputs and intermediate values, and thus prepares bounds for subsequent ReLU refinement in the next processing step.

Subsequently, the signal propagation through the ReLU activation function is considered:
\begin{align*}
        \mathbf{\hat{x}}_k &= \text{ReLU}(\mathbf{x}_k)\\
        \boldsymbol{\hat{\delta}}_k &= \text{ReLU}(\mathbf{x}_k)-\text{ReLU}(\mathbf{x}_k-\boldsymbol{\delta}_k)
    \end{align*}

The exact representation of ReLU activation functions would lead to exponential complexity. To facilitate the effective transformation of symbolic propagation mechanisms into tractable optimization problems, we adopt a systematic framework that converts neural network differential computations into verifiable constraints through convex relaxation. It is worth noting that the final error $\delta$ depends on both the tightness of the original network's normal vectors and the relaxation quality of accumulated error vectors.
Therefore, this study applies state-of-the-art convex relaxation methods~\cite{neurodiff,Zhang2018EfficientNN} to both types of vectors, achieving an optimized balance between precision and computational efficiency. Notably, computing the relaxation for both types of vectors requires their pre-activation bounds.
Given the input range $[\mathbf{l}, \mathbf{u}]$ for $\mathbf{x}$, we efficiently compute the pre-activation bounds (in layer $k$) from the inequalities $\mathbf{A}^l_k \mathbf{x}+\mathbf{b}^l_k\leq \mathbf{x}_k\leq \mathbf{A}^u_k \mathbf{x}+\mathbf{b}^u_k$ and $\mathbf{C}^l_k \mathbf{x}+\mathbf{d}^l_k\leq \boldsymbol{\delta}_k\leq \mathbf{C}^u_k \mathbf{x}+\mathbf{d}^u_k$:
\begin{align*}
    \mathbf l_{k} = {\mathbf{A}^l_k}^+\mathbf l+{\mathbf{A}^l_k}^-\mathbf u + \mathbf{b}^l_k\\
     \mathbf u_{k} = {\mathbf{A}^u_k}^+\mathbf u+{\mathbf{A}^u_k}^-\mathbf l + \mathbf{b}^u_k\\
     \underline{\mathbf d^l_{k}} = {\mathbf{C}^l_k}^+\mathbf l+{\mathbf{C}^l_k}^-\mathbf u + \mathbf{d}^l_k\\
     \overline{\mathbf d^l_{k}} = {\mathbf{C}^l_k}^+\mathbf u+{\mathbf{C}^l_k}^-\mathbf l + \mathbf{d}^u_k \label{eq:dlu}\\
     \underline{\mathbf d^u_{k}} = {\mathbf{C}^u_k}^+\mathbf l+{\mathbf{C}^u_k}^-\mathbf u + \mathbf{d}^u_k\\
     \overline{\mathbf d^u_{k}} = {\mathbf{C}^u_k}^+\mathbf u+{\mathbf{C}^u_k}^-\mathbf l + \mathbf{d}^u_k
\end{align*}
Here, $^+$ and $^-$ denote the positive and negative components of the matrices respectively.
The relaxation operates per-neuron, with $l_{k,i}$, $u_{k,i}$ representing pre-activation bounds, and $\underline{d^u_{k,i}}$, $\overline{d^u_{k,i}}$, $\underline{d^l_{k,i}}$, $\overline{d^l_{k,i}}$ denoting pre-activation error bounds for neuron $i$ in layer $k$.
We now construct \(m^l_{k,i},\ m^u_{k,i},\ n^u_{k,i},\ n^l_{k,i},\ p^l_{k,i},\ p^u_{k,i},\ q^u_{k,i},\ q^l_{k,i}\in \mathbb R\), which is a set of linear coefficients for neuron $i$ in layer $k$, for both $\hat{x}_{k,i}$ and $\hat{\delta}_{k,i}$ such that:
\begin{align*}
    m^l_{k,i}x_{k,i}+p^l_{k,i}&\leq\hat{x}_{k,i}\leq m^u_{k,i}x_{k,i}+p^u_{k,i}\\
    n^l_{k,i}\delta_{k,i}+q^l_{k,i}&\leq\hat{\delta}_{k,i}\leq n^u_{k,i}\delta_{k,i}+q^u_{k,i}
\end{align*}
The computation of these coefficient is detailed in existing works~\cite{neurodiff}.
Subsequently, we construct diagonal matrices $\mathbf{M}^u_k$, $\mathbf{M}^l_k$, $\mathbf{N}^u_k$, $\mathbf{N}^l_k$ by aggregating across neurons:
\begin{align*}
\mathbf M^u_k[i,j] &= \begin{cases}
m^u_{k,i} & \text{if } i=j \\
0 & \text{otherwise}
\end{cases} \\
\mathbf M^l_k[i,j] &= \begin{cases}
m^l_{k,i} & \text{if } i=j \\
0 & \text{otherwise}
\end{cases} \\
\mathbf N^u_k[i,j] &= \begin{cases}
n^u_{k,i} & \text{if } i=j \\
0 & \text{otherwise}
\end{cases} \\
\mathbf N^l_k[i,j] &= \begin{cases}
n^l_{k,i} & \text{if } i=j \\
0 & \text{otherwise}
\end{cases}
\end{align*}
and vectors $\mathbf{p}^u_k$, $\mathbf{p}^l_k$, $\mathbf{q}^u_k$, $\mathbf{q}^l_k$:
\begin{align*}
\mathbf p^u_k[i] = p^u_{k,i},
\mathbf p^l_k[i] = p^l_{k,i},
\mathbf q^u_k[i] = q^u_{k,i},
\mathbf q^l_k[i] = q^l_{k,i}
\end{align*}
This formulation yields the following bounding relationships:
\begin{align*}
    \mathbf M^l_k\mathbf{x}_k+\mathbf p^l_k\leq \hat{\mathbf{x}}_k\leq \mathbf M^u_k\mathbf{x}_k+\mathbf p^u_k\\
    \mathbf N^l_k\boldsymbol{\delta}_k+\mathbf q^l_k\leq \hat{\boldsymbol{\delta}}_k\leq \mathbf N^u_k\boldsymbol{\delta}_k+\mathbf q^u_k
\end{align*}

Since all elements in $\mathbf{M}^l_k$, $\mathbf{M}^u_k$, $\mathbf{N}^l_k$, $\mathbf{N}^u_k$ are non-negative (as established by~\cite{neurodiff,Zhang2018EfficientNN}), we can combine these inequalities to obtain:
\begin{align*}
\mathbf{M}^l_k(\mathbf{A}^l_{k} \mathbf{x} + \mathbf{b}^l_{k}) + \mathbf{p}^l_k &\leq \hat{\mathbf{x}}_k \leq \mathbf{M}^u_k(\mathbf{A}^u_{k} \mathbf{x} + \mathbf{b}^u_{k}) + \mathbf{p}^u_k \\
\mathbf{N}^l_k(\mathbf{C}^l_{k} \mathbf{x} + \mathbf{d}^l_{k}) + \mathbf{q}^l_k &\leq \hat{\boldsymbol{\delta}}_k \leq \mathbf{N}^u_k(\mathbf{C}^u_{k} \mathbf{x} + \mathbf{d}^u_{k}) + \mathbf{q}^u_k
\end{align*}
Through this algorithmic framework, we ultimately obtain the final error bounds at the output layer $L$.

\noindent\textbf{Justification for Dual-Network Analysis.}
While the similarity $|f(x) - f'(x)|$ could theoretically be analyzed as a single neural network $g(x) = f(x) - f'(x)$ using standard PAC-Bayes bounds or other naive single-network analysis, such approaches fundamentally ignore the strong structural correlation between $f$ and $f'$. This leads to overly pessimistic certification bounds on the deterministic error range. We justify it with the following Remark.
\begin{remark}[Structural Correlation vs. Naive Analysis]
    In compression scenarios, $W_k \approx W'_k$, meaning the weight difference $W^\Delta_k$ is often sparse or has small norms. A naive single-network analysis of $g(x)$ treats the layers as independent, which leads to loose bounds due to the accumulation of worst-case deviations (the "wrapping effect"). Our dual-network propagation (Section 4.1) linearly depends on $W^\Delta$, effectively exploiting the structural similarity to cancel out common signals and significantly tighten the range of the output error variable $\boldsymbol{\delta}_L$.
\end{remark}

\noindent\textbf{Extension to Weight Pruning via Zero-Padding. }
In weight pruning, preserved weights remain unchanged (unlike quantization-induced parameter perturbations). We optimize our error propagation model by transforming pruning elimination into parametric correction. Let $\mathcal{P}_k$ denote the set of pruned neuron indices at layer $k$ in the pruned network $f'$. For any neuron $n'_{k,i}$ in $f'$ where $i\in \mathcal{P}_k$, the deterministic relationship $\hat{x}'_{k-1,i}W'_k[i,\cdot]=0$ holds. For simplicity, we enforce $\hat{x}'_{k-1,i}=0$ while maintaining original weights $\mathbf{W'}_k$. This configuration simplifies linear error accumulation because $\mathbf{W}^\Delta_k$ and $\mathbf{b}_k-\mathbf{b}'_k$ become zero vectors by construction. Algorithm~\ref{alg:alignment} details the zero-padding alignment process for heterogeneous networks.
\begin{algorithm}[htbp]
\caption{Zero-Padding Alignment for Heterogeneous Networks}
\label{alg:alignment}
\KwIn{Original network $f$, Pruned network $f'$}
\KwOut{Aligned pruned network $\hat{f}'$}
Initialize $\hat{f}' \leftarrow f'$ \
\For{each layer $k\in{1,...,L}$}{
Obtain pruning index set $\mathcal{P}_k$ \
\If{$\mathcal{P}_k \neq \emptyset$}{
\ForEach{neuron $i \in \mathcal{P}_k$}{
Set $\hat{x}'_{k-1,i} \leftarrow 0$\tcc*[r]{Zero-pad pruned activations}
}
Maintain original weights: $\mathbf{W}'_k \leftarrow \mathbf{W}_k$ \
Maintain original biases: $\mathbf{b}'_k \leftarrow \mathbf{b}_k$ \
Set error parameters: $\mathbf{W}^\Delta_k \leftarrow \mathbf{0}$, $\mathbf{b}^\Delta_k \leftarrow \mathbf{0}$ \
}
}
\Return aligned network $\hat{f}'$;
\end{algorithm}
Error accumulation follows this linear transformation:
\begin{align*}
\boldsymbol{\delta}_k &= \boldsymbol{\hat{\delta}}_{k-1}\mathbf{W}_k
\end{align*}
For each neuron index $i$, post-ReLU error accumulation is:
\begin{align*}
\hat{\delta}_{k,i} &= \begin{cases}
\text{ReLU}(x_{k,i})& \text{if } i\in\mathcal{P}_k \\
\text{ReLU}(x_{k,i}) - \text{ReLU}(x_{k,i} - \delta_{k,i})& \text{otherwise}
\end{cases}
\end{align*}

The dual-network symbolic propagation presented in this subsection provides a linear envelope $[\delta^L, \delta^U]$ over the non-linear similarity difference $||f(x) - f'(x)||_{\infty}$. This linearization is critical: by transforming the network's complex non-linear behavior into a bounded, affine function, it enables the subsequent statistical analysis. We are now equipped to leverage this linear over-approximation to calculate the first-order (mean) and second-order (variance) moments, which are necessary inputs for deriving the variance-aware probabilistic safety guarantees using probabilistic inequality, as detailed next.
\subsection{Variance-Aware Probability Bound Computation}\label{probabilistic bound computation}

\noindent\textbf{Adapting PROVEN's Theoretical Framework}
Building on PROVEN~\cite{PROVEN}, we compute certifiably similar regions with probabilistic guarantees. This framework provides essential tools for handling unverified regions through linear output bounds. Specifically, for an unverified region $\mathcal{R}$, we obtain upper and lower bounds $g^U(\mathbf{x})$ and $g^L(\mathbf{x})$ satisfying:
\[
\forall \mathbf{x}\in \mathcal{R}, \quad g^L(\mathbf{x}) \leq g(\mathbf{x}) \leq g^U(\mathbf{x})
\]
These deterministic bounds yield probability guarantees via cumulative distribution functions (CDFs). For input random variables $\mathbf{X} \sim \mathcal{D}$:
\[
1 - F_{g^L}(\epsilon) \leq \mathbb{P}(g(\mathbf{X}) > \epsilon) \leq 1 - F_{g^U}(\epsilon)
\]
where $F_{g^U}$ and $F_{g^L}$ denote the CDF of $g^U(\mathbf{X})$ and $g^L(\mathbf{X})$ respectively.

For $\epsilon$-$\gamma$ similarity verification, we define the error $\delta(\mathbf{x}) = f(\mathbf{x}) - f'(\mathbf{x})$ and decompose the symmetric probability:
\begin{align*}
\mathbb{P}(|\delta(\mathbf{X})| \leq \epsilon) &= \mathbb{P}(\delta(\mathbf{X}) \leq \epsilon) + \mathbb{P}(-\delta(\mathbf{X}) \leq \epsilon) - 1
\end{align*}

Using linear relaxation bounds $\delta^L(\mathbf{x}) \leq \delta(\mathbf{x}) \leq \delta^U(\mathbf{x})$, we obtain the overall lower ($\gamma_{\min}$) and upper ($\gamma_{\max}$) probability bounds:
\begin{align}
\mathbb{P}(|\delta(\mathbf{X})| \leq \epsilon) &\geq F_{\delta^U}(\epsilon) + F_{-\delta^L}(\epsilon) - 1 \label{eq:prob_lower_bound} \\
\mathbb{P}(|\delta(\mathbf{X})| \leq \epsilon) &\leq F_{\delta^L}(\epsilon) + F_{-\delta^U}(\epsilon) - 1 \label{eq:prob_upper_bound}
\end{align}
The linear relaxations are defined as: $\delta^L(\mathbf{x}) = \mathbf{C}^L\mathbf{x} + \mathbf d^L$ and $\delta^U(\mathbf{x}) = \mathbf{C}^U\mathbf{x} + \mathbf d^U$.

To apply concentration inequalities to these linear bounds, we first compute the necessary statistical moments. Assuming independent and bounded input variables $\mathbf{X}$ with $X_i \in [l_i, u_i]$, the expected values (means) are:
\begin{align}
\mu^L &= \mathbb{E}[\delta^L(\mathbf{X})] = \sum^n_{i=1}C^L_{\cdot, i} \cdot \frac{l_i+u_i}{2} + \mathbf d^L \label{eq:mean_L} \\
\mu^U &= \mathbb{E}[\delta^U(\mathbf{X})] = \sum^n_{i=1}C^U_{\cdot, i} \cdot \frac{l_i+u_i}{2} + \mathbf d^U \label{eq:mean_U}
\end{align}
We define the maximum input range length $K = \max_{i\in[n]}(u_i-l_i)$.

\bigskip

\noindent\textbf{Hoeffding's Inequality Based Bounds}

We first leverage Hoeffding's inequality, which forms the basis of the existing probabilistic robustness frameworks, such as PROVEN.

\begin{lemma}[Hoeffding's Inequality]\label{lem:hoeffding_inequality}
Let $Z$ be a bounded random variable with range length $K$. The tail probability is bounded by
$$\mathbb{P}(Z - \mathbb{E}[Z] \ge t) \le \exp\left(-\frac{2t^2}{K^2}\right).$$
\end{lemma}
(Proof is omitted as this is a standard result in concentration theory.)

Applying this inequality to our linear error models yields the baseline probability bounds required for \eqref{eq:prob_lower_bound} and \eqref{eq:prob_upper_bound}:

\begin{theorem}[Hoeffding-Based Probabilistic Bound]\label{thm:hoeffding_bound_application}
Using the Hoeffding inequality (Lemma \ref{lem:hoeffding_inequality}), the four constituent CDF components are bounded as follows:

\begin{equation}
        F_{\delta^U}(\epsilon)\geq\begin{cases}
        1-\exp\left(-\frac{2(\epsilon-\mu^U)^2}{K^2\left \| \mathbf C^U \right \|^2_2 }\right) &\text{if } \epsilon-\mu^U\geq0\\
        0 &\text{otherwise}
    \end{cases} \label{eq:hoeffding_F_deltaU}
\end{equation}
and
\begin{equation}
        F_{-\delta^L}(\epsilon)\geq\begin{cases}
        1-\exp\left(-\frac{2(\epsilon+\mu^L)^2}{K^2\left \| \mathbf C^L \right \|^2_2 }\right) &\text{if } \epsilon+\mu^L\geq0\\
        0 &\text{otherwise}
    \end{cases} \label{eq:hoeffding_F_minus_deltaL}
\end{equation}
Similarly, the upper bound estimates are:
\begin{equation}
        F_{\delta^L}(\epsilon)\leq\begin{cases}
        \exp\left(-\frac{2(\epsilon-\mu^L)^2}{K^2\left \| \mathbf C^L \right \|^2_2 }\right) &\text{if } \epsilon-\mu^L\leq0\\
        1 &\text{otherwise}
    \end{cases} \label{eq:hoeffding_F_deltaL}
\end{equation}
and
\begin{equation}
        F_{-\delta^U}(\epsilon)\leq\begin{cases}
        \exp\left(-\frac{2(\epsilon+\mu^U)^2}{K^2\left \| \mathbf C^U \right \|^2_2 }\right) &\text{if } \epsilon+\mu^U\leq0\\
        1 &\text{otherwise}
    \end{cases} \label{eq:hoeffding_F_minus_deltaU}
\end{equation}
\end{theorem}

\begin{proof}
The proof involves applying the right-tail and left-tail bounds of Lemma \ref{lem:hoeffding_inequality} to the linear bounds $\delta^U(\mathbf{X})$ and $\delta^L(\mathbf{X})$. The squared range $K^2$ in the denominator is replaced by $K^2\left \| \mathbf C \right \|^2_2$.

\noindent\textbf{Derivation of Lower Bound $F_{\delta^U}(\epsilon)$ (Eq. \ref{eq:hoeffding_F_deltaU}):}
We bound the complementary probability $\mathbb{P}(\delta^U(\mathbf{X}) > \epsilon)$. Applying the right-tail bound of Lemma \ref{lem:hoeffding_inequality} with $Z=\delta^U(\mathbf{X})$ and $t = \epsilon - \mu^U$. If $t \ge 0$:
$$
F_{\delta^U}(\epsilon) = 1 - \mathbb{P}(\delta^U(\mathbf{X}) > \epsilon) \ge 1 - \exp\left(-\frac{2(\epsilon-\mu^U)^2}{K^2\left \| \mathbf C^U \right \|^2_2}\right)
$$
The derivation for $F_{-\delta^L}(\epsilon)$ (Eq. \ref{eq:hoeffding_F_minus_deltaL}) follows analogously using $\mathbb{E}[-\delta^L(\mathbf{X})] = -\mu^L$.

\noindent\textbf{Derivation of Upper Bound $F_{\delta^L}(\epsilon)$ (Eq. \ref{eq:hoeffding_F_deltaL}):}
We bound the left tail $F_{\delta^L}(\epsilon) = \mathbb{P}(\delta^L(\mathbf{X}) \le \epsilon)$. This is equivalent to applying the right-tail bound of Lemma \ref{lem:hoeffding_inequality} to the mirrored variable $Z' = \mu^L - \delta^L(\mathbf{X})$, with $\mathbb{E}[Z'] = 0$. We set $t = \mu^L - \epsilon$. If $t \ge 0$:
$$
F_{\delta^L}(\epsilon) = \mathbb{P}(Z' \ge t) \le \exp\left(-\frac{2(\mu^L-\epsilon)^2}{K^2\left \| \mathbf C^L \right \|^2_2}\right)
$$
The derivation for $F_{-\delta^U}(\epsilon)$ (Eq. \ref{eq:hoeffding_F_minus_deltaU}) follows analogously using $\mathbb{E}[-\delta^U(\mathbf{X})] = -\mu^U$.
\end{proof}

\bigskip

\noindent\textbf{Variance-Aware Bounds via Bernstein's Inequality}

The Hoeffding bound is variance-agnostic. Since the error in network similarity ($\delta$) is expected to exhibit minimal variance due to the high correlation between $f$ and $f'$, we introduce Bernstein's inequality to achieve tighter bounds.

\begin{lemma}[Bernstein's Inequality]\label{lem:bernstein_inequality}
Let $Z$ be a bounded random variable with variance $\sigma^2$ and maximum deviation $M = \max|Z - \mathbb{E}[Z]|$. The tail probability is bounded by
$$\mathbb{P}(Z - \mathbb{E}[Z] \ge t) \le \exp\left(-\frac{t^2}{2\sigma^2 + 2Mt/3}\right).$$
\end{lemma}
(Proof is omitted as this is a standard result in concentration theory.)

\begin{proposition}[Tightness Condition]\label{prop:tightness}
The Bernstein bound (Lemma \ref{lem:bernstein_inequality}) is strictly tighter than the Hoeffding bound (Lemma \ref{lem:hoeffding_inequality}) when the variance satisfies $\sigma^2 < \frac{K^2}{4} - \frac{Mt}{3}$.
\end{proposition}

\begin{proof}
The result is obtained by direct comparison of the exponents in the two inequalities. The Bernstein bound is tighter if its denominator is smaller than that of the Hoeffding bound: $2\sigma^2 + \frac{2}{3} Mt < \frac{2t^2}{t^2/K^2} = \frac{K^2}{2}$. Rearrangement yields the required condition: $\sigma^2 < \frac{K^2}{4} - \frac{1}{3} Mt$.
\end{proof}

This proposition reveals the precise conditions under which the variance-aware Bernstein inequality provides tighter bounds than the variance-agnostic Hoeffding inequality. The effectiveness of the Bernstein bound hinges on three factors:

\begin{itemize}
    \item \textbf{Low Variance ($\sigma^2$):} The most critical factor. If the error variance ($\sigma^2$) is small—a common scenario in similarity certification where the linear bounds $\delta^L$ and $\delta^U$ closely approximate the true error $\delta$—the Bernstein bound's denominator ($2\sigma^2 + \frac{2}{3} Mt$) is significantly reduced compared to Hoeffding's, leading to a tighter probability guarantee.
    \item \textbf{Tight Linear Relaxation ($M$):} $M$ represents the maximum deviation of the linear bound from its mean. A smaller $M$ (indicating a higher quality, tighter linear relaxation) minimizes the correctional term $\frac{2}{3} Mt$, further reducing the Bernstein denominator and enhancing tightness.
    \item \textbf{Small Tail Distance ($t$):} $t$ is the distance from the mean to the confidence threshold ($\epsilon-\mu$). When certifying probabilities close to the mean (small $t$), the $M$ term becomes negligible, and the bound is dominated by the variance term ($2\sigma^2$), where Bernstein achieves its maximum advantage over Hoeffding.
\end{itemize}

In summary, the Bernstein bound is particularly well-suited for similarity verification because the inherent low variance of the linear error bounds, combined with a potentially small target $\epsilon$ (i.e., small $t$), ensures the condition $\sigma^2 < \frac{K^2}{4} - \frac{Mt}{3}$ is frequently met, yielding superior certification guarantees.

\begin{theorem}[Variance-Aware Probabilistic Bound via Bernstein's Inequality]\label{thm:bernstein_bound_application}
Leveraging the variance sensitivity of Bernstein's inequality (Lemma \ref{lem:bernstein_inequality}), we derive the following tighter bounds. Here, $\text{Var}(\delta^U)$ is the variance, and $K_U$ is the maximum deviation of $\delta^U(\mathbf{X})$.

\begin{equation}
F_{\delta^U}(\epsilon)\geq\begin{cases}
1-\exp\left( -\frac{(\epsilon - \mu^U)^2}{2\text{Var}(\delta^U) + \frac{2}{3} K_U (\epsilon - \mu^U)} \right) &\text{if } \epsilon-\mu^U\geq0\\
0 &\text{otherwise}
\end{cases} \label{eq:bernstein_F_deltaU}
\end{equation}
and
\begin{equation}
F_{-\delta^L}(\epsilon)\geq\begin{cases}
1-\exp\left( -\frac{(\epsilon + \mu^L)^2}{2\text{Var}(-\delta^L) + \frac{2}{3} K_{-L} (\epsilon + \mu^L)} \right) &\text{if } \epsilon+\mu^L\geq0\\
0 &\text{otherwise}
\end{cases} \label{eq:bernstein_F_minus_deltaL}
\end{equation}
Similarly, the upper bound estimates are:
\begin{equation}
F_{\delta^L}(\epsilon)\leq\begin{cases}
\exp\left( -\frac{(\epsilon - \mu^L)^2}{2\text{Var}(\delta^L) + \frac{2}{3} K_L (\mu^L - \epsilon)} \right) &\text{if } \epsilon-\mu^L\leq0\\
1 &\text{otherwise}
\end{cases} \label{eq:bernstein_F_deltaL}
\end{equation}
and
\begin{equation}
F_{-\delta^U}(\epsilon)\leq\begin{cases}
\exp\left( -\frac{(\epsilon + \mu^U)^2}{2\text{Var}(-\delta^U) + \frac{2}{3} K_{U} (-\mu^U - \epsilon)} \right) &\text{if } \epsilon+\mu^U\leq0\\
1 &\text{otherwise}
\end{cases} \label{eq:bernstein_F_minus_deltaU}
\end{equation}
where $K_{-L}$, $K_U$ are the maximum deviations of the respective linear functions, and $\text{Var}(-\delta) = \text{Var}(\delta)$.
\end{theorem}

\begin{proof}
The derivation follows the same four-step structure as Theorem \ref{thm:hoeffding_bound_application}, but explicitly utilizes the variance-aware bound provided by Lemma \ref{lem:bernstein_inequality}.

\noindent\textbf{Derivation of Lower Bound $F_{\delta^U}(\epsilon)$ (Eq. \ref{eq:bernstein_F_deltaU}):}
We bound the right tail $\mathbb{P}(\delta^U(\mathbf{X}) > \epsilon)$ using Lemma \ref{lem:bernstein_inequality} with $Z=\delta^U(\mathbf{X})$, $\sigma^2 = \text{Var}(\delta^U)$, $M = K_U$, and $t = \epsilon - \mu^U$. If $t \ge 0$:
$$
\mathbb{P}(\delta^U(\mathbf{X}) > \epsilon) \le \exp\left( -\frac{(\epsilon - \mu^U)^2}{2\text{Var}(\delta^U) + \frac{2}{3} K_U (\epsilon - \mu^U)} \right)
$$
Thus, $F_{\delta^U}(\epsilon) = 1 - \mathbb{P}(\delta^U(\mathbf{X}) > \epsilon) \ge 1 - \exp\left( -\frac{(\epsilon - \mu^U)^2}{2\text{Var}(\delta^U) + \frac{2}{3} K_U (\epsilon - \mu^U)} \right)$. The derivation for $F_{-\delta^L}(\epsilon)$ follows analogously.

\noindent\textbf{Derivation of Upper Bound $F_{\delta^L}(\epsilon)$ (Eq. \ref{eq:bernstein_F_deltaL}):}
We bound the left tail $F_{\delta^L}(\epsilon) = \mathbb{P}(\delta^L(\mathbf{X}) \le \epsilon)$ by applying the right-tail bound of Lemma \ref{lem:bernstein_inequality} to the mirrored variable $Z' = \mu^L - \delta^L(\mathbf{X})$. We set $t = \mu^L - \epsilon$. If $t \ge 0$:
$$
F_{\delta^L}(\epsilon) = \mathbb{P}(Z' \ge t) \le \exp\left( -\frac{(\mu^L - \epsilon)^2}{2\text{Var}(\delta^L) + \frac{2}{3} K_L (\mu^L - \epsilon)} \right)
$$
The derivation for $F_{-\delta^U}(\epsilon)$ follows analogously using the left-tail bound.
\end{proof}

\bigskip

\noindent\textbf{Overall Functionality}
Our probabilistic framework offers two complementary modes for verifying neural network similarity, translating the theoretical bounds from Theorem \ref{thm:hoeffding_bound_application} and Theorem \ref{thm:bernstein_bound_application} into practical certification tools.

\noindent\textit{Probability-Bound Computation.} Given an error tolerance $\epsilon$ and input region $\mathcal{R}$, this mode computes probabilistic guarantees for network similarity. Using the derived bounds, it outputs the interval:
\begin{align}
    \underbrace{F_{\delta^U}(\epsilon) + F_{-\delta^L}(\epsilon) - 1}_{\gamma_{\min}} \leq \mathbb{P}(|\delta(\mathbf{X})| \leq \epsilon) \\
\mathbb{P}(|\delta(\mathbf{X})| \leq \epsilon)\leq \underbrace{F_{\delta^L}(\epsilon) + F_{-\delta^U}(\epsilon) - 1}_{\gamma_{\max}}
\end{align}

This provides a confidence interval $[\gamma_{\min}, \gamma_{\max}]$ quantifying the probability that output discrepancies between $f$ and $f'$ remain within $\epsilon$ over $\mathcal{R}$. The approach incorporates methodological flexibility through progressive refinement of probability intervals via domain partitioning.

\noindent\textit{Certified-Radius Computation.} For applications requiring fixed confidence levels (e.g., $\gamma = 0.9999$ for aviation systems) and reference point $\mathbf{x}_0$, this mode determines the maximum certified radius $r^*$: the largest radius satisfying
\[
\inf_{\mathbf{x} \in \mathcal{B}_\infty(\mathbf{x}_0, r^*)} \mathbb{P}(|\delta(\mathbf{X})| \leq \epsilon) \geq \gamma
\]
where $\mathcal{B}_\infty(\mathbf{x}_0, r)$ denotes the $\ell_\infty$-norm ball centered at $\mathbf{x}_0$ with radius $r$. This is achieved by solving
\[
F_{\delta^U}(\epsilon) + F_{-\delta^L}(\epsilon) - 1 = \gamma
\]
using root-finding algorithms (e.g., bisection search) over $r$ within $\mathcal{B}_\infty(\mathbf{x}_0, r)$, leveraging the functional forms of $F_{\delta^U}$ and $F_{-\delta^L}$ derived in Theorem \ref{thm:hoeffding_bound_application} or Theorem \ref{thm:bernstein_bound_application}. The resulting $r^*$ defines the operational envelope where network similarity holds with confidence $\gamma$, which is essential for deployment certification.

\section{Experimental Evaluation}\label{implement}
Our experiments are conducted on a Linux server with Ubuntu 18.04, using an Intel Xeon E5-2678 v3 processor to ensure consistent computational conditions. We implemented SimCert using Python. SimCert accepts neural network as input in Pytorch format.

\subsection{Research Questions}
This study evaluates the SimCert probabilistic verification framework for certifying behavioral similarity in neural network compression scenarios. Through systematic comparisons with PROVEN~\cite{PROVEN} and worst-case certification baselines, we address the following research questions that correspond to the two overall functionalities mentioned in Section~\ref{probabilistic bound computation}:


\begin{enumerate}
\item Does variance-aware bounding improve probability-bound computation? We examine whether SimCert yields greater width reduction ($\Delta(\gamma_U-\gamma_L) = 1 - (\gamma_U-\gamma_L)$) under identical computational constraints  across different error tolerances ($\epsilon$).

\item  In high-dimensional spaces where partitioning effectiveness diminishes, how does SimCert's certification performance compare to baselines~\cite{PROVEN,neurodiff}? We assess whether our approach certifies larger radii than PROVEN at equivalent confidence levels, and evaluate its robustness under network compression scenarios with structured pruning.
\end{enumerate}

\subsection{Experimental Results}
\noindent\textbf{ACAS Benchmark. }
The ACAS Xu benchmark~\cite{acas} comprises 45 neural networks for verifying aircraft collision avoidance systems. Each feedforward network has 5 inputs (relative distance $\rho$, relative heading $\theta$, velocities $v_{own}$ and $v_{int}$), 5 advisory outputs, and 6 hidden layers totaling 300 neurons.

Experimental results comparing width reduction using our approach (Bernstein's inequalities) versus Hoeffding's inequalities~\cite{PROVEN} are shown in Figure~\ref{width-reduction}. Across all 45 networks, our method consistently achieved greater width reduction, particularly under tighter output discrepancy budgets ($\epsilon$ = 0.01, 0.05, 0.1). Three clear patterns emerged from the data: At $\epsilon$=0.1, most measurement points exceeded the reference line; at $\epsilon$=0.05, clustering intensified above the line; and at $\epsilon$=0.01, nearly all points formed tight clusters above the line. This demonstrates that SimCert's advantage strengthens with increasingly stringent constraints.



\begin{figure*}[htbp]
\centering
\includegraphics[width=0.95\linewidth]{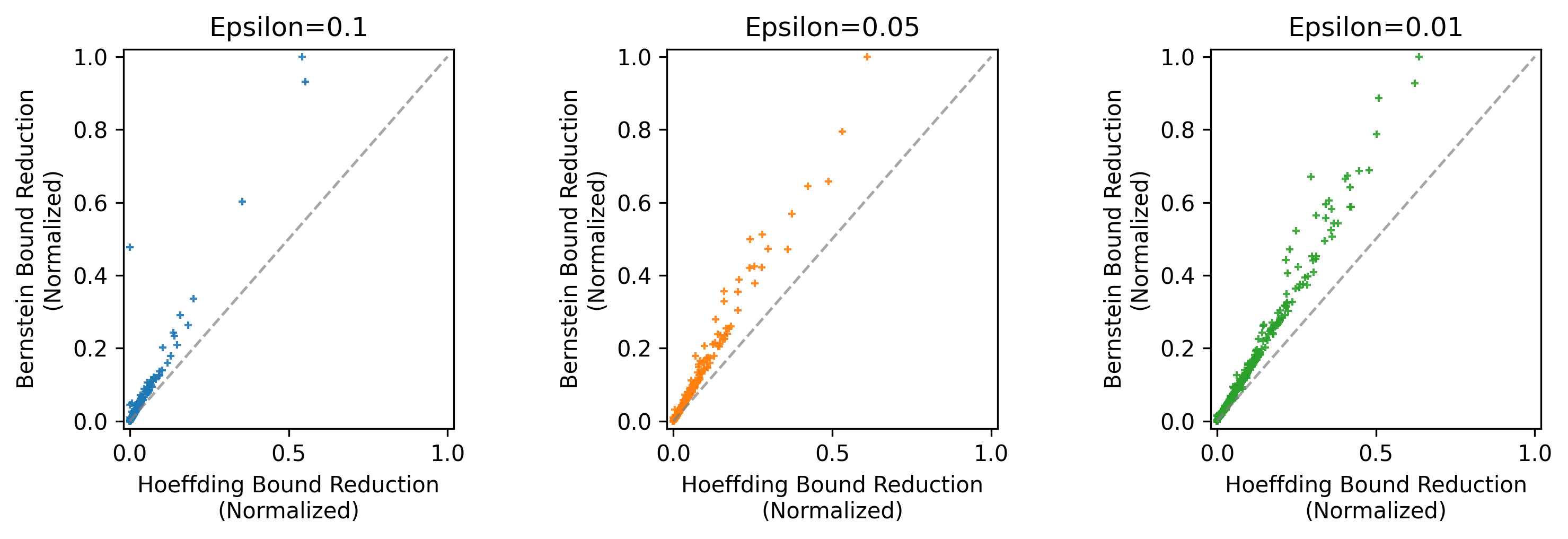}
\caption{Comparison of width reduction between Bernstein's (Our approach) and Hoeffding's (PROVEN) inequalities}
\label{width-reduction}
\end{figure*}





\noindent\textbf{MNIST. }
The MNIST handwritten digit dataset~\cite{LeCunBBH98} serves as a standard benchmark for image classification algorithms. It contains 28×28 pixel grayscale images with intensity values normalized to [0,1], providing 784 input features. Classification networks typically use an output layer with 10 nodes corresponding to digits 0-9. Final classification is determined by argmax selection.
\begin{table*}[h!]
\centering
\caption{Average certified radii comparison between PROVEN~\cite{PROVEN} and Bernstein's methods (Our approach)}
\label{tab:certification_comparison}
\begin{tabular}{lcl*{4}{S[table-format=1.5]}}
\toprule
\multirow{2}{*}{Network} &\multirow{2}{*}{Method} & \multicolumn{5}{c}{Certification Confidence}  \\
\cmidrule(lr){3-7}
 & & {5\%} & {25\%} & {50\%} & {75\%} & {99.99\%} \\
\midrule
\multirow{2}{*}{MNIST3x[1024]}
 &PROVEN& 0.09323 & 0.07939 & 0.06483 & 0.05092 & 0.02122 \\
 &Our approach& \textbf{0.15277 }& \textbf{0.13035} & \textbf{0.10604} & \textbf{0.08276} &\textbf{ 0.03241} \\
\cmidrule(r){2-7}
 &Improvement (\%)& 63.86 & 64.18 & 63.56 & 62.52 & 52.73 \\
 \midrule
\multirow{2}{*}{MNIST4x[1024]}
 &PROVEN& 0.05650 & 0.04680 & 0.03735 & 0.02896 & 0.01193 \\
 &Our approach& \textbf{0.09240} & \textbf{0.07624} & \textbf{0.06037} & \textbf{0.04619} & \textbf{0.01715} \\
\cmidrule(r){2-7}
 &Improvement (\%)& 63.54 & 62.91 & 61.63 & 59.50 & 43.76 \\
\bottomrule
\end{tabular}
\end{table*}

Probabilistic certification results under local perturbations (15 pixels, each can be perturbed from 0 to 1) are presented in Table~\ref{tab:certification_comparison} for both 3x[1024] and 4x[1024] architectures. Bernstein's method consistently outperformed Hoeffding's across all confidence levels. In the 3-layer network, certified radii increased by 63.86\% at 5\% confidence (from 0.09323 to 0.15277), with similar improvements observed at other confidence levels. The 4-layer network showed comparable gains, including a 61.63\% improvement at 50\% confidence. Notably, relative improvement decreased at extreme confidence levels, falling to 52.73\% at 99.99\% confidence compared to 63.86\% at 5\% confidence. Additionally, deeper networks exhibited smaller absolute certified radii, with a 43.05\% reduction observed at 50\% confidence when comparing the 4-layer to the 3-layer network.

\noindent\textbf{CIFAR-10. }
The CIFAR-10 dataset~\cite{krizhevsky2009learning} contains 32×32 pixel RGB images with 3 channels, providing 3072 input features normalized to [0,1]. We evaluate classification networks with 5 fully connected layers of 2048 neurons each, generating 10-dimensional logit outputs for classification via argmax selection.
This high-dimensional image recognition task tests SimCert's ability to certify behavioral similarity between original and pruned networks under structured compression.

Networks were pruned at four ratios (1\%, 5\%, 10\%, 15\%), with corresponding $\epsilon$ similarity thresholds set at 0.015, 0.075, 0.1, and 0.12 respectively. This reflects the intuitive requirement that lower pruning ratios demand stricter similarity guarantees. The input are perturbed globally (from 0 to 1) on 30 randomly selected pixels. Certification was performed at worst-case ($r_{\text{wc}}$, computed via Neurodiff~\cite{neurodiff}) and 99.99\% confidence across two methods: PROVEN~\cite{PROVEN} ($r^{P}_{0.9999}$), and our approach ($r^{B}_{0.9999}$).

Table~\ref{tab:cert_radius_grouped} reveals consistent certification patterns across compression levels, where our method ($r^{B}{0.9999}$) outperforms both alternatives at every pruning ratio. At 1\% compression ($\epsilon$=0.015), Bernstein provides 6.96\% larger radii than PROVEN ($r^{P}{0.9999}$) and 8.71\% improvement over worst-case certification ($r_{\text{wc}}$). This advantage increases to 12.54\% over Hoeffding and 20.02\% over worst-case at 5\% pruning ($\epsilon$=0.075). Crucially, at the 10\% pruning level ($\epsilon$=0.1) where all methods show reduced certification capacity, SimCert maintains a 7.89\% margin over PROVEN and 16.83\% advantage over worst-case. For the most aggressive compression (15\%, $\epsilon$=0.12), SimCert delivers 9.55\% improvement over PROVEN and 19.35\% over worst-case certification.

\begin{table*}[t]
\centering
\small
\caption{Certified radii comparison for pruned CIFAR-10 network under different compression intensities}
\setlength{\tabcolsep}{6pt}
\begin{tabular}{
  l
  S[table-format=1.2]
  S[table-format=1.4]
  S[table-format=1.6]
  S[table-format=1.6]
  S[table-format=1.6]
}
\toprule
{Network} & {pruning\_ratio} & {$\epsilon$} & {$r_{\text{wc}}$} & {$r^{P}_{0.9999}$}  & {$r^{B}_{0.9999}$}\\
\midrule
\multirow{3}{*}{5x[2048]}
  & 0.01 & 0.015 & 0.459902    & 0.467462 & \textbf{0.499966}\\
  & 0.05 & 0.075 &  0.309517   &  0.330092  & \textbf{0.371502} \\
  & 0.10 & 0.1 & 0.299285     & 0.324073  & \textbf{ 0.349674} \\
  & 0.15 & 0.12 &   0.368071     &    0.400997  & \textbf{0.439278}  \\
\bottomrule
\end{tabular}

\label{tab:cert_radius_grouped}
\end{table*}

\subsection{Discussion}
Our approach consistently demonstrated superior performance across both benchmarks, achieving greater width reduction in ACAS Xu verification and larger certified radii in MNIST and CIFAR certification. This advantage primarily stems from Bernstein's incorporation of second-order moment information, which yields tighter probability bounds for distributions with finite variance.

The experimental results reveal several important patterns. First, SimCert relative advantage amplifies under stricter constraints, as evidenced in the ACAS Xu benchmark where lower $\epsilon$ values necessitated more input domain bisections, thereby reducing sub-region variance and enhancing SimCert's efficacy. Second, the confidence level analysis in MNIST shows SimCert's relative improvement decreases at extreme confidence levels (99.99\% versus 5\%), reflecting fundamental differences in probability bound behaviors: SimCert provides tighter bounds in low-probability tails while PROVEN exhibits slower probability decay at extreme confidence levels. Third, architectural differences impact verification outcomes, with deeper MNIST networks showing smaller certified radii potentially due to increased gradient propagation complexity amplifying perturbation effects during symbolic propagation.
These differential behaviors suggest potential value in hybrid verification approaches. The consistent outperformance of our method highlights its effectiveness when input partitioning reduces sub-domain variance, though its implementation requires consideration of specific verification contexts and network architectures.

Additionally, the CIFAR results extend our MNIST observations to higher-dimensional inputs with network structural change. Our persistent advantage (7.0-16.2\% improvement over PROVEN~\cite{PROVEN}) confirms its tighter bounds transfer to complex RGB inputs despite the 4× higher dimensionality.
The consistent probabilistic advantage over worst-case certification (minimum 15.7\% gain) highlights how our approach enables a much larger certified radii with only 0.01\% loss on the confidence level.

\section{Related Works}\label{relatedwork}

Two main formal verification strategies for deep neural networks are constraint-based verification~\cite{Katz2019TheMF,Katz2017ReluplexAE} and abstraction-based verification~\cite{Albarghouthi2021IntroductionTN}. Abstraction-based approaches~\cite{2018AI,Singh2019AnAD,2018Fast} generally prioritize scalability over completeness. Techniques such as interval analysis and bound propagation~\cite{2018Efficient,2018Formal} represent specific forms of state abstraction. CROWN~\cite{Zhang2018EfficientNN} extends bound propagation from value-based bounds to linear function-based neuron bounds, improving scalability but sacrificing verification completeness. Wang et al.~\cite{Wang2021BetaCROWNEB} address this trade-off using Branch-and-Bound (BaB) with GPU acceleration, maintaining completeness while enhancing scalability.

Existing differential verification tools include ReluDiff~\cite{reludiff} and NeuroDiff~\cite{neurodiff}. These tools differ primarily in handling ReLU-induced differential outputs: ReluDiff concretizes nonlinear ReLU states, whereas NeuroDiff employs linear convex relaxations for differential output bounds, achieving greater precision. Both adopt forward verification methodologies and share a fundamental limitation, that they provide no diagnostic insights when verification returns SAT (satisfiable) results.

Probabilistic verification serves as a complementary approach to qualitative methods. Mainstream frameworks extend deterministic verification probabilistically, exemplified by PROVEN~\cite{PROVEN}. This system generalizes worst-case robustness verification to probabilistic settings, formalizing stability certification under specified input noise distributions. Empirical results demonstrate PROVEN's effectiveness over traditional certification metrics, though it remains restricted to single-network architectures.
Recent work applies probabilistic verification to neural network compression scenarios. \cite{DowningEDLBKB24} employ symbolic execution and model counting to verify user-defined robustness properties quantitatively. Their approach reports both satisfaction status and the number of violating inputs. \cite{MarzariCCF23} propose \#DNN-Verification to quantify safety violations by counting unsafe inputs, introducing both an exact counter and the probabilistic CountingProVe approximation. \cite{HuangYWDWJB14} formally verify equivalence between quantized and floating-point neural networks. To address scalability limitations in baseline MILP solvers, they develop a two-tier parallel verification framework with property-based partitioning that leverages HPC clusters.

\section{Conclusion}\label{conclusion}
This work addresses the challenge of providing formal guarantees for behavioral similarity in compressed neural networks, which is essential for deploying resource-efficient models in safety-critical applications. We have formally defined the similarity certification problem and identified two main obstacles: architectural heterogeneity from compression techniques like pruning and quantization, and limitations in existing qualitative verification approaches.
To address these challenges, we introduced SimCert, a probabilistic certification framework that integrates dual-network symbolic propagation with variance-aware probabilistic bounds using Bernstein's inequality. This combination enables precise modeling of error accumulation across networks with different architectures while providing tighter certification regions than traditional approaches.

Looking forward, our current implementation focuses on feedforward networks, but future work will extend probabilistic verification to dynamic architectures such as RNNs and Transformers. We also plan to broaden similarity certification to emerging compression techniques like knowledge distillation. This work facilitates more reliable deployment of efficient deep learning systems in resource-constrained environments through rigorous safety guarantees.

%
%

\bibliographystyle{splncs04}
\bibliography{reference}
\end{document}